\documentclass[oneside,english]{amsart}
\usepackage[T1]{fontenc}
\usepackage[latin9]{inputenc}
\usepackage{amsthm}
\usepackage{amssymb}

\makeatletter
\numberwithin{equation}{section}
\numberwithin{figure}{section}
\theoremstyle{plain}
\newtheorem{thm}{Theorem}
  \theoremstyle{definition}
  \newtheorem{defn}[thm]{Definition}
 \theoremstyle{definition}
  \newtheorem{example}[thm]{Example}

\makeatother

\begin{document}

\title{NP is not AL and P is not NC is not NL is not L}

\author{KOBAYASHI, Koji}

\maketitle

\section{Overview}

This paper talk about that NP is not AL and P, P is not NC, NC is
not NL, and NL is not L. The point about this paper is the depend
relation of the problem that need other problem's result to compute
it. I show the structure of depend relation that could divide each
complexity classes.

\section{The condition to emulate the TM by using UTM}

We will begin by considering the important nature of Turing Machine
(TM) in this paper. 
\begin{defn}
\label{def:Action Configuration}I will use the term {}``Action Configuration''
to the part of the computation configuration that decide next transition.
Action configuration include the information of state, transition
function, position (head), and read memory (tape alphabet at head).
I will use the term {}``Origin Configuration'', {}``Moving Configuration'',
{}``Target Configuration'', {}``Affirm Configuration'', {}``Negate
Configuration'', and {}``Computation Progress'' to the action configuration
of start configuration, computing configuration, halting configuration,
accepting configuration, rejecting configuration. 
\end{defn}
Action configuration have all information to yield the next configuration,
therefore we can make the universal turing machine (UTM) that emulate
TM by using the information include the action configuration.
\begin{defn}
\label{def:VTM}I will use the term {}``VTM'' as {}``Virtual Turing
Machine'' to the TM that UTM emulate.\end{defn}
\begin{thm}
\label{thm:Emulation of TM by UTM}Log space is necessary and sufficient
to record action configuration.\end{thm}
\begin{proof}
The memory space that is required with the action configuration would
be as follows: state and transition function is determined for each
TM, and TM can record in the constant space. And position can record
with log space. And read memory can record with constant space. Therefore,
UTM can record TM's action configuration into log space.
\end{proof}
Next, I talk about the sharing of the information of VTM.
\begin{thm}
\label{thm:Sharing of the information of VTM}In NTM, VTM that execute
nondeterministic branches does not share information and result each
other. If VTM share the information and result, VTM must be executed
in same branch.\end{thm}
\begin{proof}
VTM that execute non deterministic branches converges to single VTM.
The VTM is one of these branches. Another branches do not exist, and
these VTM can not affects the single VTM. Therefore, VTM that execute
non deterministic branches does not share information and result each
other.\end{proof}
\begin{thm}
\label{thm:Independency of VTM}The VTM's moving configuration that
execute in parallel must be recorded in different space. VTM that
need each other's information and result needs to execute in parallel.\end{thm}
\begin{proof}
If VTM is recorded in the other VTM space whether deterministic or
nondeterministic, UTM have to overwrite with the last VTM. So the
VTM that was overwritten another VTM can not keep the moving configuration
(especially head position) and can not continue computation. If UTM
emulates some VTM in same space, the predecessor VTM can not use successor
VTM information (like tail-recursive.) So the VTM's moving configuration
that execute in parallel must be recorded in different space. 
\end{proof}
I will define TM in this paper as follows;
\begin{defn}
I will use the term {}``NTM'' to the Nondeterministic Turing Machine
that can compute NP problems. I will use the term {}``LATM'' to
the Logarithmic space Alternating Turing Machine that can compute
AL problems. I will use the term {}``LNTM'' to the Logarithmic space
Nondeterministic Turing Machine that can compute NL problems. I will
use the term {}``LDTM'' to the Logarithmic space Deterministic Turing
Machine that can compute L problems. 
\end{defn}
To simplify, I will define UTM and VTM in this paper as follows;
\begin{defn}
\label{def:Limitations of the TM}Tape alphabet of TM is $\left\{ 0,1\right\} $.
Input data of TM is $w$. TM treats $w$ with special tape and head,
and TM does not write $w$. Length of $w$ is $O\left(n\right)$.
UTM write I will use the term {}``Working Memory'' to the memory
that TM can read and write. TM write number of the steps, therefore
computation history is acyclic. TM treats decide problems and TM must
halt.\end{defn}
\begin{thm}
\label{thm:Structure of TM}We think about the set that's elements
are target configurations. DTM's computation history is singleton,
NTM's computation history is set, ATM's computation history is family.
And structure of TM is well-founded set.\end{thm}
\begin{proof}
To think about the relation TM's computation history and result. Because
Computation history have no cyclic path in this paper, computation
history become directed acyclic graph (DAG). This DAG have root as
origin configuration, trunk as moving configuration, leaf as target
configuration. We can characterize each TM by using the DAG of the
computation history. Therefore, we can associate TM with set that
correspond with the DAG of TM's computation history. And the set is
well-founded set which minimal elements are target configurations
because DAG have no cyclic part.

DTM's computation history is only one path and have only one target
configuration. Therefore, DTM's computation history correspond with
singleton of the target configuration.

NTM's computation history is DAG. But target configuration that included
DAG affect to the NTM's result, and DAG structure does not affect
to the NTM's result. Therefore, NTM's computation history correspond
with set of the target configuration.

ATM's computation history is DAG. And DAG correspond with hypergraph
that edge correspond with universal state and existential state. Therefore,
ATM's computation history correspond with family of the target configuration's
set.
\end{proof}

\section{The depend relation between some problems}

Think the situation that some VTM is sharing the results. The problem
that describe incomplete and need the another problem's result to
complete meets the condition.
\begin{defn}
\label{def:Variable Problem}\label{def:Blocking Problem}The problem
$P_{i},P_{j}$, if $P_{i}$ value does not confirmed until $P_{j}$
value is determining, I will use the term {}``Variable Problem''
to the $P_{i}$, and {}``Blocking Problem'' to the $P_{j}$. And
I will use {}``$P_{j}P_{i}$'' to the problem that compute $P_{i}$
after computed $P_{j}$. The value or some condition of $P_{j}P_{i}$
is {}``$P_{j}P_{i}!$''. If I assume a certain value or some condition
of $P_{j}$, I will use {}``$P_{j}?$''. I will use {}``$\left[P_{i}\right]$''
to the some blocking problem of $P_{i}$, and {}``$\left[P_{i}\right]P_{i}$''
to the problem that compute $P_{i}$ after computed $\left[P_{i}\right]$. 

Furthermore, $\left[P_{i}\right]$ may be variable problem. the case
that $\left[P_{i}\right]$ is variable problem, $\left[P_{i}\right]P_{i}$
is also variable problem. The blocking problem of $\left[P_{i}\right]P_{i}$
is $\left[\left[P_{i}\right]\right]=\left[P_{i}\right]^{2}$.

\label{def:Combined Problem}I will use the term {}``Combined Problem''
and {}``$CP$'' to the issues covered in the following discussion.
Combined problem is the problem that combines some variable problems
in a complexity class. I will use the term {}``Element Problem''
and {}``$CP=\left\{ P_{0},P_{1},\cdots,P_{k-1}\right\} $'' to the
variable class. I will use {}``k'' to the total number of element
problems. Satisfiability of $P$ decide the value of $CP$. The combined
problem's value is the satisfaction of the element problems. I will
use the term {}``$V$'' to the truth value assignment of $CP$.
And I add number to each truth value assignment like $VT=\left\{ V^{0},V^{1},\cdots,V^{2^{k-1}}\right\} $.\end{defn}
\begin{example}
Parity problem of Blocking problems' true or false is variable problem.
These are four type, true is even, true is odd, false is even, false
is odd.\end{example}
\begin{defn}
\label{def:Depend Relation}\label{def:Depend Path}I will use the
term {}``Depend Relation'' and {}``$\left[P_{i}\right]\rightarrow P_{i}$''
to the relation of $\left[P_{i}\right]P_{i}$. And I will use the
term {}``Depend Path'' and {}``$\left[P_{i}\right]^{n}\rightsquigarrow P_{i}$''
to the transitive depend relation $\left[P_{i}\right]^{n}\rightarrow\left[P_{i}\right]^{n-1}\rightarrow\cdots\rightarrow P_{i}$,
and {}``$\left\{ \left[P_{i}\right]^{n}\rightsquigarrow P_{i}\right\} $''
to the set of the problems that include $\left[P_{i}\right]^{n}\rightsquigarrow P_{i}$.
For simplicity, the depend path is partial order.

I will use the term {}``Rotate Path'' to $P_{i}\rightsquigarrow P_{i}$.
And I will use {}``$\left[P_{i}\right]^{n}?\left\{ \left[P_{i}\right]^{n}\rightsquigarrow P_{i}\right\} !$''
to the computation that assume $\left[P_{i}\right]^{n}?$ and compute
$\left[P_{i}\right]^{n}\rightsquigarrow P_{i}$ and $P_{i}!$.

I will use the term {}``Depend Path Length'' and {}``$L\left(\left[P_{i}\right]^{n}\rightsquigarrow P_{i}\right)$''
to the maximum number of the depend relations in the single chain
of $\left[P_{i}\right]^{n}\rightsquigarrow P_{i}$.\end{defn}
\begin{thm}
\label{thm:Depend Relation and Motion Configuration}VTM that compute
$P_{i}!$ share the result of the VTM that compute $\left[P_{i}\right]!$.
If UTM can not record value of $\left[P_{i}\right]!$, UTM must execute
$\left[P_{i}\right]!$ VTM and $P_{i}!$ VTM in pararrel. And UTM
can not record $\left[P_{i}\right]!$ VTM and $P_{i}!$ VTM into the
same space. \end{thm}
\begin{proof}
If UTM can not record all $\left[P_{i}\right]!$, VTM must compute
$\left[P_{i}\right]!$ when compute $P_{i}!$. $P_{i}!$ need $\left[P_{i}\right]!$
and $\left[P_{i}\right]!$ need the timing to compute $\left[P_{i}\right]!$.
Therefore, $\left[P_{i}\right]!$ and $P_{i}!$ is necessary to share
the information each other.\end{proof}
\begin{thm}
\label{thm:CP as set}We can treat $CP$ as the family of the family
$P$ of the set $V$ of the $P$ that value is true. And $CP$ is
not well-formed set because of cyclic of transitive relation.\end{thm}
\begin{proof}
If we decide $\left[P_{i}\right]?$ to some $V?$, $V?$ is $P_{i}?=P_{i}!$
or $P_{i}?\neq P_{i}!$. Therefore, $P_{i}$ classify $V$ into $P_{i}?=P_{i}?P_{i}!$
or $P_{i}?\neq P_{i}?P_{i}!$. If we define $V$ as the set that include
$P_{i}?=\top$ of $\left[P_{i}\right]?$, and $P_{i}$ as the set
that include $V$ of $P_{i}?=P_{i}?P_{i}!$, $CP$ is the family of
$P_{i}$.
\end{proof}
For simplification, I will define $CP$ as follows. $P$ is the part
of rotate path. $CP$ is efficient and do not have redundant. Therefore,
all $P$ has $V$ that only $P$ is conflict. And $CP$ like $P_{i}\in CP\not\ni P_{j}\in\left[P_{i}\right]$
is exist. Such $CP$ have no limitation with $P_{j}$, therefore $CP$
can take $P_{j}!=\top$ or $P_{j}!=\bot$.

\section{$NP\supsetneq AL=P$}

Using the problem that's all part depends on whole, I show $NP\supsetneq AL=P$.
\begin{defn}
\label{def:CHAOS}I will use the term {}``CHAOS'' to the combined
problem that made the following element problems.

$P_{i}\in ClassNP$

$\left[P_{i}\right]=CP$
\end{defn}
I prove $NP\supsetneq AL$ by using CHAOS with $NP\ni CHAOS$ and
$AL\not\ni CHAOS$. 
\begin{thm}
$NP\ni CHAOS$\end{thm}
\begin{proof}
NTM can compute CHAOS to choose $P_{i}?$ in nondeterministic and
check $\forall i\left(\left[P_{i}\right]?P_{i}!=P_{i}?\right)$. And
UTM use $O\left(n\right)$ time to compute the choose of $P_{i}?$
and $P_{i}$, and compute $P_{i}!$ and compare $P_{i}?$ and $P_{i}!$.
So $NP\ni CHAOS$.
\end{proof}
I extend CHAOS and prove $AL\not\ni CHAOS$.
\begin{thm}
\label{thm:Structure of CHAOS}If we treat $CP$ as mentioned above
\ref{thm:CP as set}, CHAOS is the problem that decide $\bigcap CP=P_{0}\cap P_{1}\cap\cdots\cap P_{k-1}\neq\emptyset\leftrightarrow CP\in CHAOS$.
And CHAOS is not well-formed set.\end{thm}
\begin{proof}
If we treat $CP$ as set, $V$ that include $P$ means consistent
value $P?$. And if all $P$ include same $V$, $CP$ consistent at
$V$. Therefore, $CP$ satisfy $\bigcap CP=P_{0}\cap P_{1}\cap\cdots\cap P_{k-1}\neq\emptyset$.
And condition of CHAOS can not remove the cyclic of $CP$, therefore
CHAOS is not well-formed set.\end{proof}
\begin{thm}
$AL\not\ni CHAOS$\end{thm}
\begin{proof}
We assume that LATM can compute the CHAOS. But the assumption contradict
with CHAOS and we can see $AL\not\ni CHAOS$. 

From assumptions, there is a mapping from CHAOS to LATM. But this
mapping must relate CHAOS and LATM by using LDTM. Therefore, composition
of LDTM and LATM (of computation history) must make CHAOS structure.
And as mentioned above \ref{thm:Structure of TM}, LDTM and LATM is
well-formed, therefore CHAOS structure made from LDTM and LATM must
be well-formed.

But as mentioned above \ref{thm:Structure of CHAOS}, CHAOS is not
well-formed. If we want to treat CHAOS as well-formed structure, we
must treat some elements as minimal element and remove the cyclic
of transitive relation. And CHAOS does not include minimal elements,
LDTM must create minimal elements and LDTM or LATM must record these
elements. To remove the cyclic of transitive relation, we can use
two ways a) all $P$ include $V$ change to $P?$, and b) all $V$
include $P$ change to $V?$. a) need space as $P$ cardinality $k=\sqrt{n}>\lg\left(n\right)$.
b) need space as power set of $P$ cardinality $2^{\sqrt{n}}>\lg\left(n\right)$.
LDTM and LATM does not have a) or b) space and can not remove the
cyclic. Therefore, we can not make CHAOS structure by using composition
of LDTM and LATM.

From the above, the assumption that LATM can compute CHAOS contradict
with LATM and LDTM condition. Therefore, we can say from the reductio
ad absurdum that LATM can not compute CHAOS, and $AL\not\ni CHAOS$.\end{proof}
\begin{thm}
$NP\supsetneq AL$\end{thm}
\begin{proof}
$NP\ni CHAOS$, $AL\not\ni CHAOS$, and $NP\supset P=AL$, thus we
see $NP\supsetneq AL=P$.
\end{proof}

\section{$AL=P\supsetneq NC$}

Using the problem that's linear order structure, I show $NP\supsetneq AL$.
\begin{defn}
\label{def:ORDER}I will use the term {}``ORDER'' to the CHAOS that
made the following element problems.

$P_{i}\in ClassP$

$\left[P_{i\neq0}\right]=\left\{ P_{j}\mid j<i\right\} $
\end{defn}
I prove $P\supsetneq NC$ by using CHAOS with $P\ni ORDER$ and $NC\not\ni ORDER$. 
\begin{thm}
$P\ni ORDER$\end{thm}
\begin{proof}
UTM can compute ORDER by using this operation; both case of $P_{0}?=1$
and $P_{0}?=0$, UTM compute $\left[P_{i}\right]P_{i}!$ from smaller
number, and check $P_{0}?\left\{ P_{0}\rightsquigarrow P_{0}\right\} !=P_{0}?$.
And UTM use $O\left(n\right)$ time and space to compute all $\left[P_{i}\right]P_{i}!$.
So $P\ni ORDER$.\end{proof}
\begin{thm}
$NC\not\ni ORDER$\end{thm}
\begin{proof}
If $\left[P_{i}\right]!$ is variable, $\left[P_{i}\right]P_{i}$
is also variable problem and $\left[P_{i}\right]P_{i}!$ is variable.
If UTM compute each $P_{i}!$ in parallel, UTM must assume the combination
of $\left[P_{i}\right]?$. But $\left[P_{i}\right]?$ is reached to
$O\left(2^{n}\right)$ and UTM can not record into $O\left(n\right)$
space. And UTM must compute $\left[P_{i}\right]!$ to save the computing
space whenever $P_{i}!$ need $\left[P_{i}\right]!$. But UTM must
compute $P_{i}!$ sequentially from smaller numbers. So UTM can not
compute $P_{i}!$ in paralell.

From the above, $NC\not\ni ORDER$.\end{proof}
\begin{thm}
$P\supsetneq NC$\end{thm}
\begin{proof}
$P\ni ORDER$, $NC\not\ni ORDER$, and $P\supset NC$, thus we see
$P\supsetneq NC$.
\end{proof}

\section{$NC\supsetneq NL$}

Using the problem that's partial order structure, I show $NC\supsetneq NL$.
\begin{defn}
\label{def:LAYER}I will use the term {}``LAYER'' to the ORDER that
made the following element problems.

$P_{i}\in ClassNC$

$m>1,\, length=\left(\lg\left(n\right)\right)^{m},\, width=\dfrac{n}{length}$

$\left\{ P\right\} _{p}=\left\{ P_{q}\mid q\leqq width\times p\right\} $

$\left[P_{0}\right]=\left\{ P\right\} _{j\neq0},\left[P_{i\neq0}\right]=\left\{ P\right\} _{j<\left\lfloor \frac{i}{width}\right\rfloor }$
\end{defn}
I prove $NC\supsetneq NL$ by using CHAOS with $NC\ni LAYER$ and
$NL\not\ni LAYER$.
\begin{thm}
$NC\ni LAYER$\end{thm}
\begin{proof}
LAYER is the problem that have $width=O\left(\frac{n}{\left(\lg\left(n\right)\right)^{m}}\right)$
size anti chain of variable problem, and have $length=O\left(\left(\lg\left(n\right)\right)^{m}\right)$
length rotate path. Each variable problem in anti chain is independent
each other and UTM can compute these problems in parallel. Therefore
UTM that have $O\left(\frac{n}{\left(\lg\left(n\right)\right)^{m}}\right)<O\left(n\right)$
TM can compute LAYER in $O\left(\left(\lg\left(n\right)\right)^{m}\right)$
time.

From the above, $NC\ni LAYER$.\end{proof}
\begin{thm}
$NL\not\ni LAYER$\end{thm}
\begin{proof}
We assume that LNTM can compute the LAYER. But the assumption contradict
with LAYER and we can see $NL\not\ni LAYER$. 

In LAYER, LNTM must use $\left[P_{i}\right]?=\left\{ P\right\} _{j<\left\lfloor \frac{i}{width}\right\rfloor }?$
to compute $P_{i}!$. But LNTM can not record all $\left[P_{i}\right]?$
into $O\left(\lg\left(n\right)\right)$ space. Therefore, LNTM must
divide $\left[P_{i}\right]?$ to fit $O\left(\lg\left(n\right)\right)$
space.

But LNTM must need the information of divided $\left[P_{i}\right]?$
combination because $P_{i}!$ is changed by the $\left[P_{i}\right]?$
combination. LNTM can not use universal state, therefore LNTM must
record the information of each $\left[P_{i}\right]?$ combination.
And $\left[P_{i}\right]^{2},\left[P_{i}\right]^{3},\left[P_{i}\right]^{4},\cdots$
will also like $\left[P_{i}\right]$ and LNTM can not stop until round
rotate path. Therefore, LNTM must record at least $length=O\left(\left(\lg\left(n\right)\right)^{m}\right)$
space.

From the above, the assumption that LNTM can compute LAYER contradict
with LNTM's condition. Therefore, we can say from the reductio ad
absurdum that LNTM can not compute LAYER, and $NL\not\ni LAYER$.\end{proof}
\begin{thm}
$NC\supsetneq NL$\end{thm}
\begin{proof}
$NC\ni LAYER$, $NL\not\ni LAYER$, and $NC\supset NL$, thus we see
$NC\supsetneq NL$.
\end{proof}

\section{$NL\supsetneq L$}

Using the problem that relation spread to whole, I show $NC\supsetneq NL$.
\begin{defn}
\label{def:TWINE}I will use the term {}``TWINE'' to the LAYER that
made the following element problems.

$P_{i}\in ClassNL$

$\left[P_{0}\right]\subset\left\{ P\right\} _{j\neq0},\left[P_{i\neq0}\right]\subset\left\{ P\right\} _{j<\left\lfloor \frac{i}{width}\right\rfloor },\left|\left[P_{i}\right]\right|=O\left(\lg\left(n\right)\right)$

$O\left(L\left(P_{0}\rightsquigarrow P_{0}\right)\right)>O\left(1\right)$
\end{defn}
I prove $NL\supsetneq L$ by using CHAOS with $NL\ni TWINE$ and $L\not\ni TWINE$.
\begin{thm}
$NL\ni TWINE$\end{thm}
\begin{proof}
LNTM can compute TWINE following procedure.

First, LNTM choose $\left[P_{0}\right]$ by nondeterministic that
satisfies $\left[P_{0}\right]P_{0}?=1$ . If $\left[P_{0}\right]!$
is not exist, LNTM choose $\left[P_{0}\right]P_{0}?=0$ by nondeterministic.
If $\left[P_{0}\right]!$ is not also exist, LNTM accept input. If
$\left[P_{0}\right]!$ is exist, LNTM choose $P_{i}\in\left[P_{0}\right]$
and choose $\left[P_{i}\right]!$ by nondeterministic that satisfies
previous $\left[P_{0}\right]!$ condition. If $\left[P_{i}\right]!$
is not exist, LNTM choose $\left[P_{0}\right]P_{0}?=0$ by nondeterministic.
If $\left[P_{i}\right]!$ is not also exist, LNTM accept input. If
$\left[P_{i}\right]!$ is exist, LNTM repeat same procedure to $P_{0}$.
If LNTM reach to $P_{0}$, LNTM check $P_{0}?=P_{0}!$. If $P_{0}?=P_{0}!$
then LNTM accept input, $P_{0}?\neq P_{0}!$ in case $P_{0}?=1$ and
$P_{0}?=0$, LNTM reject input.

Such procedure, LNTM can verify all possible combinations of $P_{i}!$.
Because LNTM can verify whether all blocking problem of $P_{0}?$.
The case of $P_{i}$ is three case, a) $P_{i}!$ is the value that
never possible value of $P_{i}$, b) all $P_{i}!$ of any depend path
is same value, c) some $P_{i}!$ of depend path is different values
each other. In case a), the depend path is never exist and LNTM can
accept the branch. In case b), the depend path is correct constraint
and LNTM can continue computing. In case c), the same depend path
take true and false because the different $P_{i}!$ leads different
$\left[P_{0}\right]!$, and rotate paht will contradict at $P_{0}!$
or never possible value that refer a). Therefore LNTM can compute
correctly in a)b)c).

And this procedure use $O\left(\log\left(n\right)\right)$ space because
LNTM use one $P_{i}!$ nondeterministic and compare $P_{0}?=P_{0}!$.
From the above, $NL\ni TWINE$.
\end{proof}
I prove following lemma, and $L\not\ni TWINE$.
\begin{thm}
\label{thm:Symmetry of Combined Problem}If Combined Problem is true,
all rotate path is symmetric about satisfiability. In other words,
Decision of the Combined Problem is true, include the decision of
these rotate path is symmetric about satisfiability.\end{thm}
\begin{proof}
If Combined Problem is true, all rotate path is satisfied and symmetric
about satisfiability. Therefore, it is possible to determine whether
these rotate path is symmetric about satisfiability by determine the
true that Combined Problem.\end{proof}
\begin{thm}
\label{thm:Asymmetry of Rotate Path}The rotate path of Combined Problem
is not necessarily symmetric about satisfiability.\end{thm}
\begin{proof}
As you can see easily that is possible to create rotate path with
true and false result at same problem. Therefore, it is possible to
create rotate path that is asymmetry each other, and the rotate path
of Combined Problem is not necessarily symmetric about satisfiability.\end{proof}
\begin{thm}
\label{thm:Limit of LDTM}LDTM can handled elements atmost $O\left(n\right)$.
Therefore, LDTM can check elements symmetry or asymmetry atmost $O\left(n\right)$.\end{thm}
\begin{proof}
In order to tell apart each element, LDTM need the information. LDTM
can tell apart each element by using the pointer. But LDTM can use
atmost $O\left(\lg\left(n\right)\right)$ space, LDTM can tell apart
atmost $O\left(n\right)$ elements. Therefore, LDTM can handled elements
atmost $O\left(n\right)$.

And to check the symmetry of two elements, it's necesary to tell apart
these elements. Therefore, LDTM can check elements symmetry or asymmetry
atmost $O\left(n\right)$.\end{proof}
\begin{thm}
\label{thm:Symmetry of TM}When dealing with a Combined Problem, NTM
can deal with the symmetry of the elements in same step. But DTM can
not deal with the symmetry of the elements in same step.\end{thm}
\begin{proof}
When computing a Combined Problem, DTM have at most one computation
history that is one way from starting configuration to halting configuration.
DTM's computation configuration can not replace another. And DTM can
not deal some elements symmetry at each step.

But NTM have branching computation history that is Directed Acyclic
Graph which root is starting configuration. Therefore, some branches
that have same trunk is symmetry and can replace each other. And NTM
can deal some element symmetry by dealing these element as branches.\end{proof}
\begin{thm}
\label{thm:Size of TWINE's Rotate Path}In TWINE, number of different
sequences of values in a rotate path is $O\left(n^{L\left(P_{0}\rightsquigarrow P_{0}\right)}\right)>O\left(n^{c}\right)$.\end{thm}
\begin{proof}
In TWINE, number of different sequences of values $\left[P_{i}\right]$
is atmost $O\left(n\right)$, because $\left|\left[P_{i}\right]\right|=\lg\left(n\right)$.
And because of TWINE's structure, length of rotate path is atmost
$L\left(P_{0}\rightsquigarrow P_{0}\right)>O\left(1\right)$. Therefore,
number of different sequences of values in a rotate path is $O\left(\overset{L\left(P_{0}\rightsquigarrow P_{0}\right)}{\prod}\left[P_{i}\right]\right)=O\left(n^{L\left(P_{0}\rightsquigarrow P_{0}\right)}\right)>O\left(n^{c}\right)$.\end{proof}
\begin{thm}
$L\not\ni TWINE$\end{thm}
\begin{proof}
We assume that LDTM can compute the TWINE. But the assumption contradict
with CHAOS and we can see $L\not\ni TWINE$. 

First, We think that compute rotate path. Proof. As mentioned above\ref{thm:Symmetry of Combined Problem},
all rotate path symmetry in satisfiability if TWINE is true. Thus
computing that TWINE is true include that all rotate path is symmetry.
And as mentioned above\ref{thm:Asymmetry of Rotate Path}, the rotate
path of TWINE is not necessarily symmetric about satisfiability, LDTM
must compute to compare their satisfiability. And as mentioned above\ref{thm:Symmetry of TM},
DTM can not deal some symmetry, DTM must deal these rotate path separately.

As mentioned above\ref{thm:Size of TWINE's Rotate Path}, number of
rotate path is $O\left(n^{L\left(P_{0}\rightsquigarrow P_{0}\right)}\right)>O\left(n^{c}\right)$.
As mentioned above\ref{thm:Limit of LDTM}, LDTM can check rotate
path symmetry or asymmetry atmost $O\left(n\right)$, and can not
check all rotate path. Therefore, LDTM must use multiple LDTM to check
all rotate path symmetry.

For checking the symmetry of rotate path, LDTM must tell apart each
rotate path. LDTM can handle each element atmost $O\left(n\right)$.
Therefore, LDTM must split all rotate path to fit $O\left(n\right)$.
The number of the rotate path pack are $O\left(\dfrac{n^{L\left(P_{0}\rightsquigarrow P_{0}\right)}}{n}\right)=O\left(n^{L\left(P_{0}\rightsquigarrow P_{0}\right)-1}\right)$.
LDTM can check symmetry all rotate path to check these pack. But LDTM
can not tell apart each rotate path pack, LDTM must repeat thus splitting
$O\left(L\left(P_{0}\rightsquigarrow P_{0}\right)\right)$ times.

We think the number of required LDTM to split rotate path. LDTM must
split rotate path and execute sub LDTM to check symmetry, and finally
check each sub LDTM's result and each symmetry. I will use the term
{}``Caller LDTM'' to the LDTM that split rotate path and execute
sub LDTM, and {}``Callee LDTM'' to the LDTM that called by Caller
LDTM. Callee LDTM must get the rotate path pack information to check
the symmetry from Caller LDTM. Caller LDTM must get the result information
from Callee LDTM. Therefore, as mentioned above\ref{thm:Independency of VTM},
Caller LDTM and Callee LDTM must execute in parallel and must use
different space.

Thus chain from Caller LDTM to Callee LDTM exist $O\left(L\left(P_{0}\rightsquigarrow P_{0}\right)\right)>O\left(1\right)$.
Constant LDTM can not compute these chain. That is inconsistent with
assumptions and thus can not compute with LDTM.

From the above, $L\not\ni TWINE$.\end{proof}
\begin{thm}
$NL\supsetneq L$\end{thm}
\begin{proof}
$NL\ni TWINE$, $L\not\ni TWINE$, and $NL\supset L$, thus we see
$NL\supsetneq L$.
\end{proof}

\section{Conclusion}

These results lead to the conclusion.
\begin{thm}
$NP\supsetneq AL=P\supsetneq NC\supsetneq NL\supsetneq L$\end{thm}

\end{document}